\documentclass[aps,10.4pt,twocolumn, twoside]{revtex4}
\usepackage{amsthm}
\usepackage{amsmath,latexsym,amssymb,verbatim,enumerate,graphicx,bbm}
\usepackage{color}

\newtheorem{definition}{Definition} 
\newtheorem{prop}[definition]{Proposition}
\newtheorem{lemma}[definition]{Lemma}

\newtheorem*{rep@theorem}{\rep@title}
\newcommand{\newreptheorem}[2]{%
\newenvironment{rep#1}[1]{%
 \def\rep@title{#2 \ref{##1} (restatement)}%
 \begin{rep@theorem}}%
 {\end{rep@theorem}}}
\makeatother

\newreptheorem{thm}{Theorem}
\newreptheorem{lem}{Lemma}

\def\ba#1\ea{\begin{align}#1\end{align}}
\def\ban#1\ean{\begin{align*}#1\end{align*}}

\newcommand{\be}{\begin{equation}}
\newcommand{\ee}{\end{equation}}

\newcommand{\md}{\mathrm{d}}
\newcommand{\me}{\mathrm{e}}

\def\benum{\begin{enumerate}}
\def\eenum{\end{enumerate}}

\def\squareforqed{\hbox{\rlap{$\sqcap$}$\sqcup$}}
\def\qed{\ifmmode\squareforqed\else{\unskip\nobreak\hfil
\penalty50\hskip1em\null\nobreak\hfil\squareforqed
\parfillskip=0pt\finalhyphendemerits=0\endgraf}\fi}
\def\endenv{\ifmmode\;\else{\unskip\nobreak\hfil
\penalty50\hskip1em\null\nobreak\hfil\;
\parfillskip=0pt\finalhyphendemerits=0\endgraf}\fi}

\newcommand{\bra}[1]{\langle #1|}
\newcommand{\ket}[1]{|#1\rangle}

\newcommand{\tr}{\text{tr}}

\newcommand{\<}{\langle}
\renewcommand{\>}{\rangle}

\def\be{\begin{equation}}
\def\ee{\end{equation}}
\def\ben{\begin{eqnarray}}
\def\een{\end{eqnarray}}

\def\bei{\begin{itemize}}
\def\eei{\end{itemize}}

\mathchardef\ordinarycolon\mathcode`\:
\mathcode`\:=\string"8000
\def\vcentcolon{\mathrel{\mathop\ordinarycolon}}
\begingroup \catcode`\:=\active
  \lowercase{\endgroup
  \let :\vcentcolon
  }

\newcommand{\nc}{\newcommand}
 \nc{\proj}[1]{|#1\rangle\!\langle #1 |} 
\nc{\avg}[1]{\langle#1\rangle}

\nc{\conv}{\operatorname{conv}}
\nc{\smfrac}[2]{\mbox{$\frac{#1}{#2}$}} \nc{\Tr}{\operatorname{Tr}}
\nc{\ox}{\otimes} \nc{\dg}{\dagger} \nc{\dn}{\downarrow}
\nc{\lmax}{\lambda_{\text{max}}}
\nc{\lmin}{\lambda_{\text{min}}}

\nc{\csupp}{{\operatorname{csupp}}}
\nc{\qsupp}{{\operatorname{qsupp}}} \nc{\var}{\operatorname{var}}
\nc{\rar}{\rightarrow} \nc{\lrar}{\longrightarrow}
\nc{\poly}{\operatorname{poly}}
\nc{\polylog}{\operatorname{polylog}} \nc{\Lip}{\operatorname{Lip}}
\nc{\Om}{\Omega}
\nc{\wt}[1]{\widetilde{#1}}

\def\>{\rangle}
\def\<{\langle}

\nc{\glneq}{{\raisebox{0.6ex}{$>$}  \hspace*{-1.8ex} \raisebox{-0.6ex}{$<$}}}
\nc{\gleq}{{\raisebox{0.6ex}{$\geq$}\hspace*{-1.8ex} \raisebox{-0.6ex}{$\leq$}}}

\nc{\vholder}[1]{\rule{0pt}{#1}}
\nc{\wh}[1]{\widehat{#1}}
\nc{\h}[1]{\widehat{#1}}

\nc{\ob}[1]{#1}

\def\beq{\begin {equation}}
\def\eeq{\end {equation}}

\def\be{\begin{equation}}
\def\ee{\end{equation}}

\nc{\eq}[1]{(\ref{eq:#1})} 
\nc{\eqs}[2]{\eq{#1} and \eq{#2}}

\nc{\eqn}[1]{Eq.~(\ref{eqn:#1})}
\nc{\eqns}[2]{Eqs.~(\ref{eqn:#1}) and (\ref{eqn:#2})}

\nc{\region}{\cS\cW}

\begin{document}

\title{{\Large Entanglement area law from specific heat capacity}}

\author{Fernando G.S.L. Brand\~ao}
\email{f.brandao@ucl.ac.uk}
\affiliation{Department of Computer Science, University College London}

\author{Marcus Cramer}
 \email{marcus.cramer@uni-ulm.de}
\affiliation{Institut fÃ\"ur Theoretische Physik, Universit\"at Ulm, D-89069 Ulm, Germany}


\begin{abstract}

We study the scaling of entanglement in low-energy states of quantum many-body models on lattices of arbitrary dimensions. We allow for unbounded Hamiltonians such that  systems with bosonic degrees of freedom are included. We show that if at low enough temperatures the specific heat capacity of the model decays exponentially with inverse temperature, the entanglement in every low-energy state satisfies an area law (with a logarithmic correction).
This behaviour of the heat capacity is typically observed in gapped systems. Assuming merely that the low-temperature specific heat decays polynomially with temperature, we find a subvolume scaling of entanglement. Our results give experimentally verifiable conditions for area laws, show that they are a generic property of low-energy states of matter, and, to the best of our knowledge, constitute the first proof of an area law for unbounded Hamiltonians beyond those that are integrable.





\end{abstract}

\maketitle

\parskip .75ex


The amount of entanglement in low-energy states of quantum many-body models has been the subject of intense examination. The problem was originally studied in relation to the Bekenstein entropy formula for blackholes \cite{Bek71, Haw74, BKLS86} and more recently also in the context of condensed-matter physics and quantum information theory \cite{ECP10, AEPW02, VLRK03, CC04, PEDC05, CEP07, Wol06}. The behaviour of entanglement in physically relevant states is an interesting topic not only due to the resource character of entanglement in quantum information \cite{HHHH08, BBPS96}, but also because it can be used to elucidate aspects of the physics of the system \cite{ECP10, Wol06}. Another motivation comes from the observation that systems with a large amount of entanglement are usually hard to simulate classically. As a consequence it is useful to identify when there is only limited entanglement in the system. Indeed it turns out that in many circumstances a small amount of entanglement leads to good ways to simulate the physics of the model numerically \cite{Has07, VC06, FNW92, Whi92}. 

Given a bipartite pure state $\ket{\psi}_{AB}$, the entanglement entropy of $A$ with $B$ is given by \cite{HHHH08, BBPS96}
\begin{equation}
E(\ket{\psi}_{AB}) := S(\rho_A) = - \tr(\rho_A \log \rho_A),
\end{equation}
with $\rho_A=\text{tr}_B(\rho)$ the reduced density matrix of $\ket{\psi}_{AB}$ on the region $A$ and $S(\rho_A)$ its von Neumann entropy. 

Starting with \cite{Bek71, Haw74, BKLS86} and later also with \cite{AEPW02, PEDC05, CEP07, Wol06} for quantum quasi-free systems and  \cite{VLRK03, CC04} for integrable quantum spin systems, a large body of work appeared indicating that low-energy states of local models satisfy an \textit{area law} \cite{ECP10}, i.e. the entanglement of a contiguous region with its complementary region is proportional to its boundary, in contrast to its volume as is the typical behaviour for a generic quantum state \cite{HLW06}. 

The problem is particularly well understood for groundstates of one-dimensional bounded Hamiltonians: A seminal result of Hastings \cite{Has07} (see also \cite{AKLV12}) gives an area law for the groundstate of every 1D gapped model. Recently this result was generalized to an area law for every one-dimensional state that has a finite correlation length \cite{BH12b, BH12a}. In contrast, there are groundstates of gapless 1D models (or states with diverging correlation length) with a volume scaling of entanglement \cite{Irani10, GH09}.

Systems in dimension larger than one are not nearly as well understood and, apart from solvable cases~\cite{AEPW02,ECP10,PEDC05,CEP07}, almost nothing is known for unbounded Hamiltonians as those of bosonic systems. Neither states with a finite correlation length nor, more particularly, groundstates of gapped models are known to universally obey an area law. It is thus interesting to find further conditions under which an area law can be proven, which is the approach taken in \cite{Has07c, Mas09, AMV13, Cho14, SM14}. In this paper we follow this direction and link area laws to another important property of physical systems, different from spectral gap and correlation length. Namely, we connect it to the specific heat capacity of the model at low temperatures. We show that whenever the specific heat decays with the temperature fast enough, the entanglement of \textit{every} low-energy state of the model is significantly limited. 

We focus on translationally-invariant nearest-neighbour Hamiltonians on a $d$-dimensional lattice $\Lambda := \{1, \ldots, n \}^d$ and consider non-translational local Hamiltonians acting beyond nearest-neighbours in the Appendix.
Separating on-site terms and terms coupling nearest-neighbours, the Hamiltonian reads
\begin{equation}
\label{ham}
H=\sum_{i\in\Lambda}H_{i}+\sum_{i\in\Lambda} H_{B(i)}
\end{equation}
where each $H_{i}$ acts only on site $i$ and each $H_{B(i)}$ only on $B(i)=\{j\in\Lambda\,|\,|i-j|\le1\}$. 
Without loss of generality we assume that the ground state energy is zero.
We will connect the scaling of entanglement in the ground state to the specific heat of the thermal state. To this end,
define the thermal state at temperature $T$ as $\rho_T := e^{-H/T}/Z_{T}$, where $Z_T := \tr(e^{-H/T})$ is the partition function (we set Boltzmann's constant to $1$). The energy and entropy densities are given by $u(T) := \tr(H \rho_T)/n^d$ and $s(T) := S(\rho_T)/n^d$, respectively. 
Finally, we define the specific heat capacity at temperature $T$ as
\begin{eqnarray}
c(T) := \dot u(T) =  \frac{1}{n^d T^2} \text{Cov}_{\!\rho_T}(H,H),
\end{eqnarray}
where we used the notation $\dot u(T)=\frac{\md u}{\md T}(T)$ as we will do in the following for all functions depending on temperature
and $\text{Cov}_{\!\rho}(A,B)=\langle (A-\langle A\rangle_{\!\rho})^\dagger(B-\langle B\rangle_{\!\rho})\rangle_{\!\rho}$ denotes the covariance of operators $A$ and $B$ in $\rho$.

We will need to be able to put bounds on the $H_{B(i)}$.
If these are bounded operators, we simply set $h=\|H_{B(i)}\|$. As we are allowing for infinite-dimensional Hilbert spaces, this is a bit more tricky for unbounded operators. The quantity corresponding to $h$ will turn out to be state-dependent and requires some notation. To introduce it, we consider an example first.
The Bose--Hubbard model 
\begin{equation}
H=-J\sum_{|i-j|=1}b_{i}^\dagger b_{j}+U\sum_{i}n_{i}(n_{i}-1)-\mu\sum_{i}n_{i},
\end{equation}
may be written as in Eq.~\eqref{ham} with $H_i=Un_{i}(n_{i}-1)-\mu n_{i}$ and 
\begin{equation}
\begin{split}
H_{B(i)}&=-Jb_{i}^\dagger \sum_{j:|i-j|=1}b_{j}.
\end{split}
\end{equation}
Let us consider a cubic region $R=\{1,\dots,l\}^d$ and let $i\in R$.
We may bi-partition $B(i)$ in the following way. Write $A=B(i)\cap R$ and $B=B(i)\cap (\Lambda\backslash R)$.
Then $A\subset R$ and $B\subset \Lambda\backslash R$.
If $i$ is in the ``interior'' $R^\circ=\{2,\dots,l-1\}^d$ of $R$ then $B$ is empty and $H_{B(i)}$ acts only on $R$. If, however, 
$i$ is in the ``boundary'' $\partial R=R\backslash R^\circ$ of $R$ then it also acts on sites outside of $R$.
E.g., if $i$ is as in Fig.~\ref{fig1} then $A=\{i,i-e_1,i-e_2,i+e_2\}$ and $B=\{i+e_1\}$ with $e_\delta$ the euclidean unit vectors.
We may then write
\begin{equation}
\label{blah}
\begin{split}
H_{B(i)}=\sum_{k=1}^Kh^{(k)}_{A}\otimes h^{(k)}_B,
\end{split}
\end{equation}
with $K=2$ and $h^{(1)}_{A}=-Jb_{i}^\dagger$, $h^{(1)}_B=b_{i+e_1}$, $h^{(2)}_B=\mathbbm{1}_B$, and $h^{(2)}_{A}$ collecting the remaining terms which only act on $A$. In this way, we may write $H_{B(i)}$ as in Eq.~\eqref{blah} for any $i\in\partial R$ and this is of course also possible for more general Hamiltonians than the Bose--Hubbard model. For a given translationally-invariant state $\rho$ and given Hamiltonian $H$ as in Eq.~\eqref{ham} with $H_{B(i)}$, $i\in\partial R$, as in Eq.~\eqref{blah} we may thus define
\begin{equation}
\label{op_norm_replacement}
h(\rho)=\max_{i\in\partial R}\,\Bigl|\sum_{k=1}^K\text{Cov}_{\!\rho}\bigl(h^{(k)\dagger}_{A},h^{(k)}_B\bigr)\Bigr|,
\end{equation}
where $\text{Cov}_{\!\rho}(h_A,h_B)$ is the covariance of operators $h_A$ and $h_B$ in $\rho$ as defined above and we note that, by the Cauchy Schwarz inequality, $|\text{Cov}_\rho(h_A,h_B)|^2\le \text{Cov}_{\!\rho}(h_A,h_A)\text{Cov}_{\!\rho}(h_B,h_B)$. 
E.g., for the Bose--Hubbard model one finds
\begin{equation}
h(\rho)\le 2d|J|\text{Cov}_{\!\rho}\bigl(b_i,b_i\bigr)
\le 2d|J|\langle n_i\rangle_{\!\rho}
\end{equation}
such that an upper bound is provided by the mean occupation number.

\begin{figure}[t]
\begin{center}  
\includegraphics[width=0.9\columnwidth,angle=0]{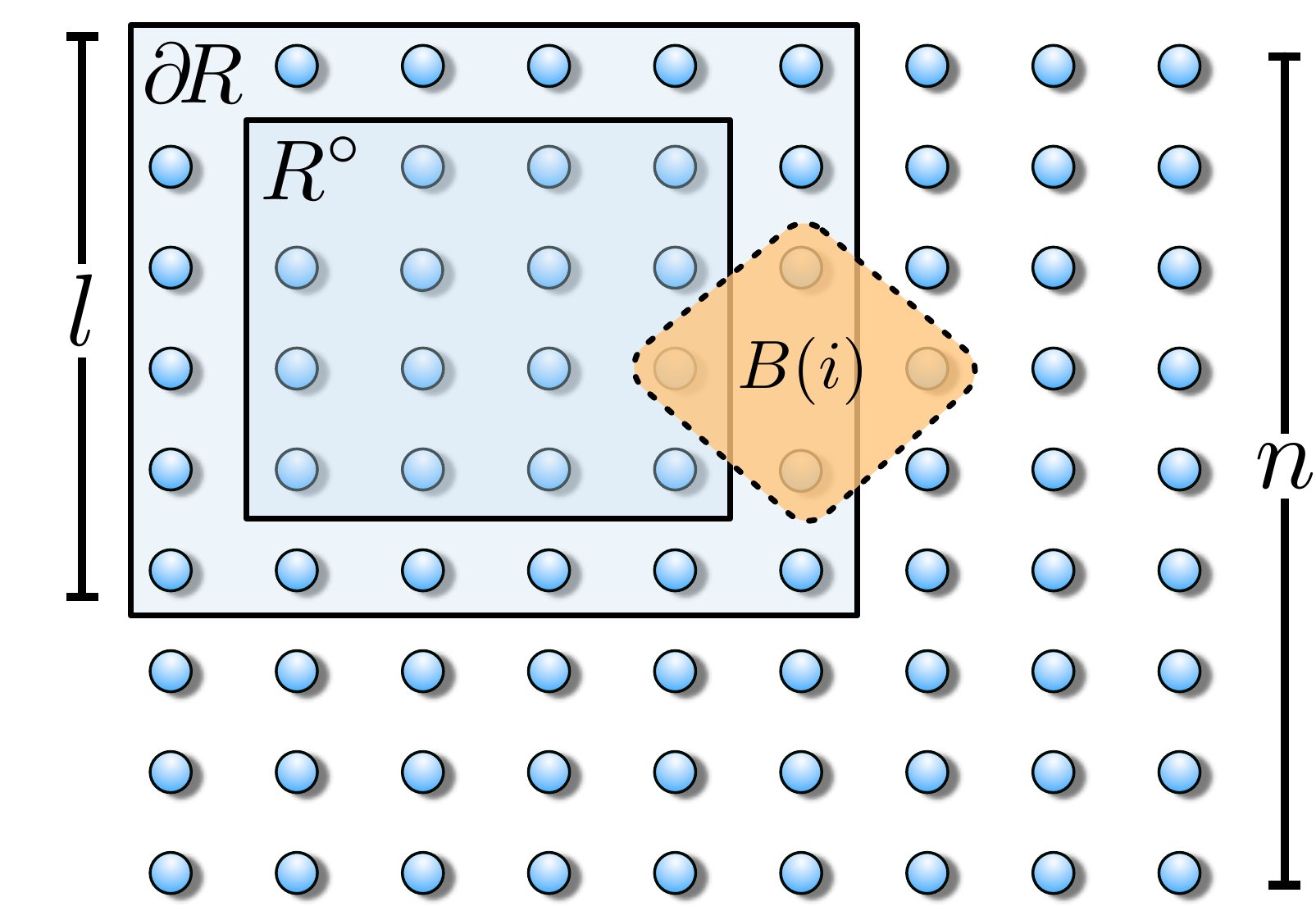}  
\caption{The cubic lattice $\Lambda=\{1,\dots,n\}^d$ under consideration. We consider Hamiltonians $H=\sum_{i\in\Lambda}(H_i+H_{B(i)})$ where each $H_i$ acts only on site $i$ and each $H_{B(i)}$ only on nearest-neighbours $B(i)$. The area law is obtained for cubic regions $R$ of edge length $l$, which may be partitioned into its boundary $\partial R$ and interior $R^\circ$. Whenever $i\in \partial R$ then $H_{B(i)}$ acts on $R^\circ$ and $\Lambda\backslash R$ as shown by the yellow region. For unbounded $H_{B(i)}$ these terms enter the area-law bound  through Eq.\ \eqref{op_norm_replacement}. \label{fig1}}  
\end{center}  
\end{figure}  

 Given a region $R \subset \Lambda$, we denote by  $\tr_{\Lambda \backslash R}$ the partial trace over all sites in $\Lambda$ except those in $R$ and write $\rho_R=\tr_{\Lambda \backslash R}(\rho)$. 
 Our main result is the following (we present the case of non-translationally-invariant local Hamiltonians that act beyond nearest-neighbours in the Appendix).

\begin{prop}  \label{thmAreaLaw}
Let $H$ be a translationally-invariant nearest-neighbour Hamiltonian on a $d$-dimensional lattice $\Lambda =\{1,\dots,n\}^d $ as in Eq.~\eqref{ham} and let $R$ be a cubic region of edge length $l$ (and volume $l^d$).
Let $H$ have groundstate degeneracy $D$ and let $\rho=|\psi\rangle\langle\psi|$ a translationally-invariant state with $\text{tr}(H\sigma)\le Cn^d/l$ for some $C\ge 1$. 
Let $T_c$  such that 
\begin{equation}
u(T_c) =  \frac{C+4dh(\sigma)}{l}.
\end{equation}
If there are $\gamma, k ,\Delta > 0$ such that for every $T \leq T_c$
\begin{enumerate}
\item 
$c(T) \leq k (\Delta/T)^\gamma \me^{- \Delta / T}$  then there is a constant $C_0$ depending only on $C,\gamma,\Delta,k,h,d$ such that
\begin{equation} \label{arealawlog}
S( \rho_R ) \leq C_{0} l^{d-1}\log(l).
\end{equation}

\item $c(T) \leq k (T/\Delta)^{\gamma}$ then there is a constant $C_0$ depending only on $C,\gamma,\Delta,k,h,d$ such that
\begin{equation}
S( \rho_R ) \leq C_{0}l^{d-1+\frac{1}{\gamma+1}}.
\end{equation}
\end{enumerate}
Here, $h=\|H_{B(i)}\|$ for bounded $H_{B(i)}$ and $h= h(\rho)$ as in Eq.~\eqref{blah} for unbounded $H_{B(i)}$.
\end{prop}

Part~1 gives an area law with logarithmic correction for the von Neumann entropy of every low-energy state, assuming that the specific heat decreases exponentially with inverse temperature at temperatures smaller than $T_c$, which decreases with increasing $l$. Part 2, in turn, shows merely a subvolume law; however only the weaker assumption $c(T) \leq k T^{\gamma}$ is required.

We expect that part~1 of the proposition can be strengthened to show a strict area law, without a logarithmic correction. However we do not have a proof and leave it as an open question. A drawback of the proposition is that one must know the behaviour of the specific heat at arbitrary small temperatures. This can be circumvented if we assume that at small temperatures the heat capacity is monotonically increasing with temperature, as is usually observed. Then it is enough that the heat capacity decays as specified in the proposition only in the range $[T_c/2, T_c]$.


Which Hamiltonians have these two types of heat capacity dependence on temperature? Although we do not have a general result concerning the question, the class of systems satisfying the required conditions appears to be very large. Indeed gapped models are expected to have $c(T) \leq T^{-\nu} e^{- \Delta / T}$ at all sufficiently small temperatures \cite{Ein1907}. Examples where this has been verified explicitly---both in theory and experimentally---are superconductors \cite{Tin04, Cor46, DBKHBM97}, quantum Hall systems \cite{SW07, MA98, CP97}, and some lattice spin systems \cite{MG07, QRMGCK07, HRB98}. The condition $c(T) \leq T^{\gamma}$ is even more general, and is the behaviour routinely observed in measurements of heat capacity of solids. It is an interesting open question to give more rigorous results in this respect.


An appealing aspect of the result is that the heat capacity is readily accessible experimentally (in fact the decay of the heat capacity of the form $T^{-\nu} e^{- \Delta / T}$ is many times used as an experimental signature that the system has a spectral gap $\Delta$; for instance the first measurement of the superconductor gap was achieved precisely by measuring this form for the specific heat \cite{Cor46}). Thus one can infer an area law even when the Hamiltonian is not fully known.

The proof of the proposition is relatively simple. It is a consequence of the variational characterization of thermal states as states of minimum free energy, together with standard thermodynamical formulas relating specific heat, energy density, and entropy density. We start with the following lemma:

\begin{lemma}  \label{thelemma}
Let $H$ on $\Lambda=\{1,\dots,n\}^d$ as above and let $R$ a cubic region of edge length $l$. For every translationally state $\rho$ and temperature $T$ such that $u(T)\ge\text{tr}(H\rho)/n^d+4dh(\rho)/l$ one has
\begin{equation}
S(\rho_R) \leq l^d s(T).
\end{equation}
Here, $h(\rho)=\|H_{B(i)}\|$ for bounded $H_{B(i)}$ and $ h(\rho)$ as in Eq.~\eqref{blah} for unbounded $H_{B(i)}$.
\end{lemma}

\begin{proof} Partition the lattice into $M=n^d/l^d$ cubic subsets $R_1,\dots,R_M$, each of edge-length $l$ and define
\begin{equation}
\sigma=\bigotimes_{m=1}^M\rho_{R_m},
\end{equation}
for which $\sigma_{R_m}= \rho_{R_m}$, i.e., $\sigma$ and $ \rho$ coincide locally on each $R_{m}$. 
As the thermal state minimizes the free energy 
$F_T(\rho) := \tr(H\rho)-TS(\rho)$
(This well-known property of the free energy is a direct consequence of the identity $F_T(\rho) = F_T(\rho_T) + S(\rho || \rho_T)$, where $S(\rho || \rho_T)$ is the relative entropy, which is always non-negative) one finds
\begin{equation}
\label{free}
T S(\sigma) \le T S(\rho_T)-\tr(H \rho_T)+\tr(H \rho) -\tr(H (\rho-\sigma)).
\end{equation}
Using translational invariance and the fact that $\rho$ and $ \sigma$ coincide on each $R_m$
\begin{equation}
\begin{split}
\tr(H(\rho-\sigma))=M\sum_{i\in \partial R_1}\tr\left(H_{B_r(i)}(\rho_{B(i)}- \sigma_{B(i)})\right)
\end{split}
\end{equation}
such that for bounded $H_{B_r(i)}$
\begin{equation}
\begin{split}
|\tr(H(\rho- \sigma))|&\le 2M\max_m|\partial R_m|h
\le 4dhMl^{d-1}.
\end{split}
\end{equation}
For unbounded $H_{B_r(i)}$, $i\in\partial R_1$, we write 
\begin{equation}
\begin{split}
\tr\left(H_{B_r(i)}(\rho- \sigma)\right)&=\tr\left(H_{B_r(i)}(\rho- \sigma_{R_1}\otimes \sigma_{\Lambda\backslash R_1})\right)\\
&=\sum_{k=1}^K\text{Cov}_{\!\rho}\left(h_A^{(k)\dagger} ,h_B^{(k)}\right).
\end{split}
\end{equation}
By the additivity of the entropy and translational invariance $MS(\sigma_{R_1})=\sum_{m=1}^MS(\sigma_{R_m})=S(\sigma)$
and by Eq.~\eqref{free} and the above bounds
\begin{equation}
\begin{split}
 S(\sigma)\le   S(\rho_T)+\frac{\tr(H \rho) -n^du(T)+4dh(\rho)Ml^{d-1}}{T}
 \end{split}
\end{equation}
such that the assertion follows by the choice of $M$ and whenever $\sigma$ and $T$ are as in the hypothesis. 
\end{proof}
We now turn to the proof of the proposition.

\begin{proof}[Proof of Proposition \ref{thmAreaLaw}]
By Lemma~\ref{thelemma} and for $\rho$ and $T_c$ as in the hypothesis
\begin{equation}
\label{Sbound}
\begin{split}
\frac{1}{l^d}S(\sigma_{R})&\le  s(T_c)=s(0)+\int_0^{T_c}\md T\,\frac{\dot u(T)}{T}.
\end{split}
\end{equation}
We now derive a bound on temperatures $T$ below $T_c$ in terms of $u(T)$.
For $c(T)$ upper bounded as in part~2
\begin{equation}
\begin{split}
\frac{u(T)}{k}=\int_0^T\!\!\!\md t\,\frac{c(t)}{k}\le  \int_0^T\!\!\!\md t\,\frac{t^\gamma}{\Delta^\gamma}=\frac{\Delta}{\gamma+1}\left(\frac{T}{\Delta}\right)^{\gamma+1},
\end{split}
\end{equation}
i.e., $T^{-1}\le C_{k,\Delta,\gamma}\dot g(u(T))$, where $g(u)=u^{\frac{\gamma}{\gamma+1}}$. Hence,
\begin{equation}
\begin{split}
\int_0^{T_c}\!\!\!\md T\,\frac{\dot u(T)}{C_{k,\Delta,\gamma}T}&\le \int_0^{T_c}\!\!\!\md T\,\frac{\md}{\md T} g(u(T))
=g(u(T_c)),
\end{split}
\end{equation}
which, when inserting the definition of $g$ and combining it with Eq.~\eqref{Sbound}, proves part~2.
Part~1 is proved in complete analogy and may be found in the Appendix.
\end{proof}

\noindent \textbf{Comparison with Previous Work:} The fact that the Gibbs state minimizes the free energy was instrumental in the proof of Proposition~1. It is not the first time that this relation is used in the context of area laws: It was employed in \cite{WFHC08} to show a general area law for the mutual information of every Gibbs state of a local Hamiltonian at constant temperature. 
Concerning groundstates, Hastings \cite{Has07c} and Masanes \cite{Mas09} proved an area law with a $O(\log(n))$ correction (logarithmic of the total system size) for gapped models, assuming that the density of states $D(E)$ grows polynomially in the volume of the lattice for energies $E \leq O(n^d)$, i.e. $D(E) \leq n^{O(E)}$. The density of states and the heat capacity are related by
\begin{eqnarray} \label{cvcD}
 n^d T^2 c(T)  &=& \frac{1}{Z_{T}}\int D(E) e^{- E/T} E^2 dE \nonumber \\ &&\hspace{0.5cm} -  \left(  \frac{1}{Z_{T}} \int D(E) e^{- E/T} E dE   \right)^2. \nonumber 
\end{eqnarray}
Therefore it is interesting to compare the results of \cite{Has07c, Mas09} with Proposition~1. Assuming $D(E) \leq n^{cE}$ for a constant $c \geq 0$ and that $H$ has spectral gap $\Delta > 0$  we find
\begin{eqnarray} \label{implication}
c(T) &\leq& \frac{1}{n^d T^2 } \int_{\Delta}^{\infty} n^{cE} e^{- E/T} E^2 dE \nonumber \\ &\leq& O(T^{-1}n^{c\Delta - d}e^{- \Delta / T}).
\end{eqnarray}
Thus, part~1 of Proposition~1 gives that $S(\tr_{\Lambda \backslash R}(\ket{\psi}\bra{\psi}) \leq O(l^{d-1}\log(n))$, recovering the area law result of \cite{Has07c, Mas09}  \footnote{We note \cite{Mas09} also shows that if the density of states of the sub-Hamiltonian on region $R$ plus a boundary of $R$ (of width $O(\log |\partial R|)$) grows polynomially with the energy (up to energies of order $|\partial R|$) and the groundstate has a finite correlation length, then it satisfies an area law with a correction of $O(\log |R|)$. }. Our argument has the advantage of being simpler than \cite{Has07c, Mas09}, which were based on high temperature expansion and Lieb-Robinson bounds, respectively. In contrast all we needed is the variational principle for the free energy and basic relations between specific heat and energy/entropy.

\noindent \textbf{PEPS Approximation:} Ref. \cite{Has07c} also proved that under the above assumption on the density of states, a constant spectral gap, and bounded Hamiltonian, the groundstate is well approximated by the Gibbs state at temperature $O(1/\log(n))$. Using the result of \cite{Has06} it was then argued that such groundstates are well approximated by a projected-pair-entangled-operator (PEPO---the mixed state analogue of projected-pair-entangled-states (PEPS)) \cite{VC06} of quasipolynomial $\exp(O(\log^d(n)))$ bond dimension. We can easily show that this is already the case under our assumption on the heat capacity: 

Under the assumptions of part~1 of Proposition~1, one can show (see Appendix) that for every $0<\delta\le 1$ and $T\le \delta/\log(n)$ the maximally mixed state $\rho_0$ on the groundspace of $H$ is well approximated by the thermal state at temperature $T$,
\begin{equation}
\|\rho_T-\rho_0\|_1\lesssim \left(\frac{\log(n)}{\delta}\right)^{\gamma-1}n^{d-\Delta/\delta},
\end{equation}
where $\Delta,\gamma$ are as in Proposition~1.
Ref. \cite{MSVC14}, which builds on \cite{Has06, KGKRE14}, shows that $\rho_T$ can be approximated with error $\varepsilon$ in trace norm by a PEPO of bond dimension $(n/\varepsilon)^{O(\log(n))}$.

\noindent \textbf{Conclusion:} We showed that specific heat dependence on temperature commonly observed in gapped systems implies an area law with logarithmic correction for every low-energy state, and, building on  \cite{Has06, KGKRE14, MSVC14}, that the groundstate can be approximated by a PEPS of small bond dimension. Even the mild assumption that the heat capacity decays polynomially already puts non-trivial constraints on the amount of entanglement in the system. We believe our results are valuable for four main reasons: First the condition on the specific heat is natural, can be checked experimentally, and relates entanglement scaling to a thermodynamic quantity; second it implies an area law not only for the groundstate, but for every state of sufficiently low energy; third it applies to unbounded Hamiltonians; finally, the argument is very simple. 
It is an interesting open question to prove a strict area law, without any logarithmic correction, under the assumption of an exponentially decreasing specific heat. Another interesting direction for future research is to elucidate the class of models having the required behaviour for the specific heat. Although examples of gapped models violating it can be constructed, it might still be possible to prove it for a large class of gapped models.


\section{Acknowledgements}

FGSLB acknowledges EPSRC and MC the EU Integrated Project SIQS and the Alexander von Humboldt foundation
 for financial support.

\appendix
\widetext

\section{More general version of Proposition 1 and Lemma 2}
We let $H$ a $r$-local Hamiltonian on a $d$-dimensional cubic lattice $\Lambda=\{1,\dots,n\}^d$. That is, 
\begin{equation}
H=\sum_{i\in\Lambda}H_{B_r(i)},
\end{equation}
where  each $H_{B_r(i)}$ acts only on sites
\begin{equation}
B_r(i)=\left\{j\in\Lambda\,\big|\,d(i,j)\le r\right\},
\end{equation}
where $d(\cdot,\cdot)$ denotes the Manhattan distance in the lattice.
Let $M=n^d/l^d$ and partition
\begin{equation}
\Lambda=\bigcup_{m=1}^MR_{m}
\end{equation}
with $R_m$ cubes of edge length $l$ and denote by 
\begin{equation}
\mathbb{E}[f(R)]=\frac{1}{M}\sum_{m=1}^Mf(R_m)
\end{equation}
the uniform expectation of $f$ over all the $R_m$. Further, $R_m=v+\{1,\dots,l\}^d$ for appropriate $v=v(m)$ and we denote the ``interior'' of $R_m$ by
$R^\circ_m=v+\{1+r,\dots,l-r\}^d$ and the ``boundary'' of $R_m$ by
$\partial R_m=R_m\backslash R^\circ_m$, for which we have $|\partial R_m|\le 2drl$. For given $i\in\partial R_m$ let $A=A(i)=B_r(i)\cap R_m$, $B=B(i)=B_r(i)\cap (\lambda\backslash R_m)$, and
\begin{equation}
H_{B_r(i)}=\sum_{k=1}^K h_{A}^{(i,k)}\otimes h_{B}^{(i,k)}.
\end{equation}
For a given state $\rho$ define
\begin{equation}
h(\rho)=\min\left\{\max_{i\in\Lambda}\|H_{B_r(i)}\|,\,\max_{i\in \partial R}\,\bigl|\sum_{k=1}^K \text{Cov}_{\!\rho}\bigl(h_{A}^{(i,k)\dagger}, h_{B}^{(i,k)}\bigr)\bigr|\right\}.
\end{equation}
For states $\rho$ and $A\subset\Lambda$ we denote $\rho_{A}=\text{tr}_{\Lambda\backslash A}[\rho]$ so the state obtained by tracing over all sites apart from those in $A$.
We have the following lemma.
\begin{lemma} If $\sigma$ and $T$ are such that $\tr(H \sigma)/n^d +4drh(\sigma)/l\le u(T)$
then $\mathbb{E} [S(\sigma_{R})]\le  l^ds(T)$.
\end{lemma}
\begin{proof}
Define
\begin{equation}
\tilde \sigma=\bigotimes_{i=1}^M\sigma_{R_i},
\end{equation}
for which $\tilde \sigma_{R_i}= \sigma_{R_i}$, i.e., $\sigma$ and $\tilde \sigma$ coincide locally on each $R_i$.
As the thermal state minimizes the free energy we have
\begin{equation}
T S(\tilde\sigma) \le T S(\rho_T)-\tr(H \rho_T)+\tr(H \tilde\sigma) =T S(\rho_T)-\tr(H \rho_T)+\tr(H \sigma) -\tr(H (\sigma-\tilde\sigma)),
\end{equation}
where, as $\sigma$ and $\tilde \sigma$ coincide on each $R_i$,
\begin{equation}
\begin{split}
\tr(H(\sigma-\tilde \sigma))&=\sum_{i\in\Lambda}\tr\left(H_{B_r(i)}(\sigma-\tilde \sigma)\right)
=\sum_{m=1}^M\sum_{i\in R_m}\tr\left(H_{B_r(i)}(\sigma-\tilde \sigma)\right)
=\sum_{m=1}^M\sum_{i\in \partial R_m}\tr\left(H_{B_r(i)}(\sigma_{B_r(i)}-\tilde \sigma_{B_r(i)})\right),
\end{split}
\end{equation}
such that for bounded $H_{B_r(i)}$
\begin{equation}
\begin{split}
|\tr(H(\sigma-\tilde \sigma))|&\le 2M\max_m|\partial R_m|\max_i\|H_{B_r(i)}\|\le 4drhn^d/l.
\end{split}
\end{equation}
For unbounded $H_{B_r(i)}$, $i\in\partial R$, we write 
\begin{equation}
\begin{split}
\tr\left(H^{(i)}_{B_r(i)}(\sigma-\tilde \sigma)\right)&=\tr\left(H^{(i)}_{B_r(i)}(\sigma- \sigma_{R_m}\otimes \sigma_{\Lambda\backslash R_m})\right)\\
&=\sum_{k=1}^K\text{Cov}_{\!\sigma}\left(h_A^{(i,k)\dagger} ,h_B^{(i,k)}\right).
\end{split}
\end{equation}
Hence, by the additivity of the entropy,
\begin{equation}
\sum_{i=1}^MS(\sigma_{R_i})=
 S(\tilde\sigma) \le  S(\rho_T)+n^d\frac{\tr(H \sigma)/n^d +4drh(\sigma)/l-u(T)}{T}
\end{equation}
and thus the assertion follows whenever $\sigma$ and $T$ are as in the hypothesis. 
\end{proof}

\begin{prop} Let $\sigma$ such that $\text{tr}(H\sigma)\le Cn^d/l$ for some $C\ge 1$ and $T_c$ such that 
\begin{equation}
u(T_c)=\frac{C+4drh(\sigma)}{l}.
\end{equation}
Then
\begin{equation}
\mathbb{E}[S(\sigma_{R})]\le s(0)l^d+F_{k,\gamma,\Delta,l}l^{d-1},
\end{equation}
where, if $c(T)\le k(T/\Delta)^\gamma$ for some $k,\Delta,\gamma>0$ and all $T\le T_c$,
\begin{equation}
F_{k,\gamma,\Delta,l}=\frac{k^{\frac{1}{\gamma+1}}}{\gamma}(\gamma+1)^{\frac{\gamma}{\gamma+1}}\left(\frac{C+4drh(\sigma)}{\Delta}\right)^{\frac{\gamma}{\gamma+1}}
l^{\frac{1}{\gamma+1}}
\end{equation}
and, if $c(T)\le k (\Delta/T)^\gamma\me^{-\Delta/T}$ for some $k,\Delta,\gamma>0$ and all $T\le T_c$,
\begin{equation}
F_{k,\gamma,\Delta,l}=
2\left(\ln\left(k\gamma^{\gamma-1}\right)+1+\gamma/2+\ln(l)\right)\frac{C+4drh(\sigma)}{\Delta}.
\end{equation}
\end{prop}
\begin{proof} By the above Lemma and for $\sigma$ and $T_c$ as in the hypothesis
\begin{equation}
\begin{split}
\frac{1}{Ml^d}\sum_{i=1}^MS(\sigma_{R_i})&\le  s(T_c)=s(0)+\int_0^{T_c}\md T\,\frac{\dot u(T)}{T}.
\end{split}
\end{equation}
We now derive a bound on temperatures $T$ below $T_c$ in terms of $u(T)$.
For $c(T)\le k(T/\Delta)^\gamma$, $\Delta,\gamma> 0$, we have
\begin{equation}
\begin{split}
u(T)&=\int_0^T\md t\,c(t)\le k \int_0^T\md t\,(t/\Delta)^{\gamma}=\frac{k\Delta}{\gamma+1}(T/\Delta)^{\gamma+1},
\end{split}
\end{equation}
i.e.,
\begin{equation}
\begin{split}
\frac{\Delta}{ T}\le\left(\frac{k\Delta}{\gamma+1}\right)^{\frac{1}{\gamma+1}}\left(u(T)\right)^{-\frac{1}{\gamma+1}}=\left(\frac{k\Delta}{\gamma+1}\right)^{\frac{1}{\gamma+1}}\frac{\gamma+1}{\gamma}\dot g(u(T)),\;\;\; g(u)=u^{\frac{\gamma}{\gamma+1}}
\end{split}
\end{equation}
such that
\begin{equation}
\begin{split}
\int_0^{T_c}\md T\,\frac{\dot u(T)}{T}&\le \frac{1}{\Delta}\left(\frac{k\Delta}{\gamma+1}\right)^{\frac{1}{\gamma+1}}\frac{\gamma+1}{\gamma}\int_0^{T_c}\md T\,  \dot g(u(T)) \dot u(T)\\
&=\frac{1}{\Delta}\left(\frac{k\Delta}{\gamma+1}\right)^{\frac{1}{\gamma+1}}\frac{\gamma+1}{\gamma}\left( g(u(T_c))-g(u(0))\right)
\\
&=\frac{1}{\Delta}\left(\frac{k\Delta}{\gamma+1}\right)^{\frac{1}{\gamma+1}}\frac{\gamma+1}{\gamma} \left(u(T_c)\right)^{\frac{\gamma}{\gamma+1}}.
\end{split}
\end{equation}

Now let $c(T)\le k(\frac{\Delta}{T})^{\gamma} \me^{-\Delta/T}$, $\Delta,k>0$, $\gamma\ge 0$.
For $0\le t\le 1/\gamma$, the function $\me^{-1/t}t^{-\gamma}$ is non-decreasing in $t$ such that for $T\le T_c\le \Delta/\gamma$ 
\begin{equation}
\begin{split}
u(T)&=\int_0^T\md t\,c(t)\le k \Delta\int_0^{\frac{T}{\Delta}}\md t\,T^{-\gamma} \me^{-1/T}\le k \Delta(T/\Delta)^{1-\gamma} \me^{-\Delta/T}\le \gamma^{\gamma-1}\left(\frac{\Delta}{\gamma T} \me^{-\frac{\Delta}{\gamma T}}\right)^{\gamma}
\le k\gamma^{\gamma-1}\me^{-\frac{\Delta}{2 T}},
\end{split}
\end{equation}
i.e., for $T\le T_c\le \Delta/\gamma$,
\begin{equation}
\label{intbound}
\begin{split}
\frac{\Delta}{2 T}
&\le \ln\left(k\gamma^{\gamma-1}\right)-\ln(u(T))=\ln\left(k\gamma^{\gamma-1}\right)+\dot g(u(T)),\;\;\; g(u)=u-u\ln(u)
\end{split}
\end{equation}
such that for $T_c\le \Delta/\gamma$
\begin{equation}
\begin{split}
\int_0^{T_c}\md T\,\frac{\dot u(T)}{T}&\le \frac{2}{\Delta}\ln\left(k\gamma^{\gamma-1}\right)u(T_c)+\frac{2}{\Delta}\left( g(u(T_c))-g(u(0))\right)
\\
&=\frac{2}{\Delta}\left(\ln\left(k\gamma^{\gamma-1}\right)+1-\ln(u(T_c))\right)u(T_c)
\end{split}
\end{equation}
and for $T_c> \Delta/\gamma$
\begin{equation}
\begin{split}
\int_0^{T_c}\md T\,\frac{\dot u(T)}{T}&=\int_0^{\frac{\Delta}{\gamma}}\md T\,\frac{\dot u(T)}{T}+\int_{\frac{\Delta}{\gamma}}^{T_c}\md T\,\frac{\dot u(T)}{T}\\
&\le \int_0^{\frac{\Delta}{\gamma}}\md T\,\frac{\dot u(T)}{T}+\frac{\gamma}{\Delta}(u(T_c)-u(\Delta/\gamma))\\
&\le \frac{2}{\Delta}\ln\left(k\gamma^{\gamma-1}\right)u(\Delta/\gamma)+\frac{2}{\Delta} \left[u(\Delta/\gamma)
- u(\Delta/\gamma)\ln(u(\Delta/\gamma))\right]
+\frac{\gamma}{\Delta}(u(T_c)-u(\Delta/\gamma))\\
&=\frac{2\ln(k\gamma^{\gamma-1})-\gamma}{\Delta}
u(\Delta/\gamma)+\frac{\gamma}{\Delta}u(T_c)\\
&\hspace{2cm}
+\frac{2}{\Delta}u(T_c) \left[\frac{u(\Delta/\gamma)}{u(T_c)}
- \frac{u(\Delta/\gamma)}{u(T_c)}\ln\left(\frac{u(\Delta/\gamma)}{u(T_c)}\right)
- \frac{u(\Delta/\gamma)\ln(u(T_c))}{u(T_c)}
\right]
\\
&\le\frac{2}{\Delta}\left(\ln(k\gamma^{\gamma-1})+1-\ln(u(T_c))
\right)
u(\Delta/\gamma)
+\frac{\gamma}{\Delta}u(T_c)-\frac{\gamma}{\Delta}u(\Delta/\gamma)
-\frac{2}{\Delta}u(\Delta/\gamma),
\end{split}
\end{equation}
where we used Eq.~\eqref{intbound} and the monotonicity  of $u(T)$ (recall that we choose $T_c$ the smallest temperature for which .. holds, i.e., below this temperature the energy is non-decreasing) to bound the first integral and the monotonicity of $x-x\ln(x)$ for $0\le x\le 1$ to obtain the last line.
\end{proof}
\section{PEPO approximation}
\begin{prop}
Let $H$ be a Hamiltonian on a $d$-dimensional lattice $\Lambda := [n]^d$. Let $\rho_0$ be the maximally mixed state on the groundspace of $H$. Suppose there are $\Delta, k  >0$ and $\nu \geq 0$ such that $c(T) \leq k T^{-\nu} e^{- \Delta / T}$ for $T \leq 1/\log(n)$. Then for every $1 \geq \delta > 0$ and $T \leq \delta/\log(n)$, 
\begin{equation} \label{boundtracenorm}
\Vert \rho_T - \rho_0 \Vert_1 \leq \eta := \frac{2k}{\Delta} \left( \frac{\log(n)}{\delta} \right)^{\nu-1} n^{d - \frac{\Delta}{\delta}}, 
\end{equation}
and there is a PEPO $\pi$ of bond dimension $(n/\varepsilon)^{O(\log(n))}$ such that $\Vert \pi - \rho_0 \Vert_1 \leq \varepsilon + \eta$.
\end{prop} 

\begin{proof}
For $T = \delta/\log(n)$ we find
\begin{eqnarray}
s(T) -s(0)=  \int_{0}^{T} \frac{C(T')}{T'}dT' \leq  k \frac{\log^\nu(n)}{\delta^{\nu}n^{ \frac{\Delta}{\delta}}} .
\end{eqnarray}
Eq. (\ref{boundtracenorm}) follows from the previous equation and
\begin{eqnarray}
&&\hspace{-1cm}\frac{1}{Z_T}\sum_{k > 0}e^{- E_k/T}  \leq \frac{T}{\Delta} \frac{1}{Z_T}\left( \sum_{k > 0} \frac{E_k}{T} e^{- E_k/T} \right)  \nonumber \\   
&=&  \frac{T }{\Delta}  (S(\rho_T) - \log(Z_T)) \leq \frac{T n^d}{\Delta}  (s(T) - s(0)).  \nonumber
\end{eqnarray}
Ref. \cite{MSVC14}, which builds on \cite{Has06, KGKRE14}, shows that $\rho_T$ can be approximated with error $\varepsilon$ in trace norm by a PEPO of bond dimension $(n/\varepsilon)^{O(\log(n))}$.
\end{proof}
\end{document}